\documentclass[sn-mathphys]{sn-jnl}

\newenvironment{bprooftree}
{\leavevmode\hbox\bgroup}
{\DisplayProof\egroup}
\usepackage{amssymb}
\usepackage{bussproofs}
\usepackage{multicol}
\usepackage{rotating}
\usepackage{microtype}      
\usepackage{lipsum}		
\usepackage{amsfonts,amsmath,amsthm}
\usepackage{verbatim}

\jyear{2022}%

\theoremstyle{thmstyleone}%
\usepackage{natbib}
\theoremstyle{thmstyletwo}%
\begin{document}

\title[Computational Paths - Computational Paths]{Computational Paths - - An approach in the $LND_{EQ}-TRS_{2}$ system.}


\author*[1]{\fnm{Tiago} \sur{M.L. de Veras}}\email{tiago.veras@ufrpe.br}

\author[2]{\fnm{Arthur} \sur{F. Ramos}}\email{arthur742@gmail.com}
\author[3]{\fnm{Ruy} \sur{J.G.B. de Queiroz}}\email{ruy@cin.ufpe.br}

\author[3]{\fnm{Anjolina} \sur{G. de Oliveira}}\email{ago@cin.ufpe.br}

\affil*[1]{\orgdiv{Mathematics Department}, \orgname{ Universidade Federal Rural de Pernambuco}, \orgaddress{\street{Rua Dom Manuel de Medeiros}, \city{Recife}, \postcode{52171-900}, \state{PE}, \country{Brasil}}}

\affil[2]{\orgname{Microsoft Redmond}, \orgaddress{\street{700 Bellevue Way NE}, \city{Bellevue}, \postcode{98004}, \state{WA}, \country{USA}}}

\affil[3]{\orgdiv{Centro de Inform\'atica}, \orgname{Universidade Federal de Pernambuco}, \orgaddress{\street{Av. Jornalista Anibal Fernandes s/n}, \city{Recife}, \postcode{50740-560}, \state{PE}, \country{Brasil}}}


\abstract{We use a labelled deduction system ( LND$_{ED-}$TRS ) based on the concept of computational paths (sequences of rewrites) as equalities between two terms of the same type, which allowed us to carry out in homotopic theory an approach using the concept of computational paths. From this, we show that the computational paths can be used to perform the proofs of the $LND_{EQ}-TRS_{2}$ rewriting system.}




\maketitle
\section{Introduction}\label{sec1}

The identity type is arguably one of the most interesting entities of  Martin-L\"{o}f type theory (MLTT). From any type $A$, it is possible to construct the identity type $Id_A (x,y)$ whose inhabitants (if any) are proofs of equality between $x$ and $y$. This type establishes the relation of identity between two terms of $A$, i.e., if there is a construction $x =_p y: A$, then $p$ is a witness or proof that $x$ is indeed equal to $y$. Both $Id_A (x,y)$ and $x =_p y: A$ are types of equality, the second being an explicit way of defining equality. The proposal of the Univalence Axiom made the identity type perhaps one of the most studied aspects of type theory in the last decade or so. It proposes that in type theory, to say $x=y$ is equivalent to saying that $x\simeq y$, that is,the identity type is equivalent to the type of equivalences. Another important aspect is the fact that it is possible to interpret the paths between two points of the same space. This interpretation gives rise to the interesting view of equality as a collection of homotopical paths. And such connection of type theory and homotopy theory makes type theory a suitable foundation for both computation and mathematics. Nevertheless, in the original formulation of intensional identity type in MLTT, this interpretation is only a semantical one 
 and it was not proposed with a syntactical counterpart for the concept of path in type theory.

    For that reason, the addition of paths to the syntax of homotopy type theory has been recently proposed by \cite{Ruy1,Art3,Art4,Art5}, in these works, the authors use an entity known as `computational path', proposed by \cite{Ruy4}, and show that it can be used to formalize the identity type in a more explicit manner.

On the other hand, one of the main interesting points of the interpretation of logical connectives via deductive systems which use a labelling system is the clear separation between a functional calculus on the labels/terms (the names that record the steps of the proof) and a logical calculus on the formulas \cite{Lof1},\cite{Ruy4}. Moreover, this interpretation has important applications. The works of \cite{Ruy1},\cite{Ruy4},\cite{Ruy5},\cite{RuyAnjolinaLivro} show that the harmony that comes with this separation makes labelled natural deduction a suitable framework to study and develop a theory of equality for natural deduction. Take, for example, the following cases taken from the $\lambda$-calculus \cite{Ruy5}:

\begin{center}
$(\lambda x.(\lambda y.yx)(\lambda w.zw))v \rhd_{\eta} (\lambda x.(\lambda y.yx)z)v \rhd_{\beta} (\lambda y.yv)z \rhd_{\beta} zv$

$(\lambda x.(\lambda y.yx)(\lambda w.zw))v \rhd_{\beta} (\lambda x(\lambda w.zw)x)v \rhd_{\eta} (\lambda x.zx)v \rhd_{\beta} zv$
\end{center}

In the theory of the $\beta\eta$-equality of $\lambda$-calculus, we can indeed say that $(\lambda x.(\lambda y.yx)(\lambda w.zw))v$ is equal to $zv$. Moreover, as we can see above, we have at least two ways of obtaining these equalities. We can go further, and call $s$ the first sequence of \textit{rewrites} that establish that $(\lambda x.(\lambda y.yx)(\lambda w.zw))v$ is indeed equal to $zv$. The second one, for example, we can call $r$. Thus, we can say that this equality is establish by $s$ and $r$. As we will see in this paper, $s$ and $r$ are examples of an entity known as \textit{computational paths}. 

Since we now have terms (alongside formulas) which are supposed to formalise sequences of rewrites (i.e., computational paths) establishing the equality between two path-terms, interesting questions might arise: Is $s$ different from $r$, or are they normal forms of this equality proof? If $s$ is equal to $r$, how can we prove this? We can answer questions like this when we work in a labelled natural deduction framework. The idea is that we are not limited by the calculus on the formulas, but we can also define and work with rules that apply to the labels/terms. That way, we can use these rules to formally establish the equality between these labels, i.e., establish equalities between equalities. In this work, we will use a system proposed by \cite{Anjo1} and known as $LND_{EQ}$-$TRS$.

By introducing a framework that formalizes the logical notion of equality through identity type, Martin-L\"{o}f type theory allows a surprising connection between the rewriting of labeled terms and homotopy theory. As a matter of fact, MLTT allows for making useful bridges between theory of computation, algebraic topology, logic, categories, and higher algebra, and a single concept seems to serve as a bridging bond: \emph{path}. The impact in mathematics has been felt more strongly since the start of Vladimir Voevodsky's program on the univalent foundations of mathematics around 2005, and one specific aspect which we would like to mention here is the calculation of fundamental groups of spaces. Most algebraic topology textbooks, when bringing in the definition of ``homotopy group", draw the attention to the fact that the calculation of homotopy groups is in general much more difficult than some of the other homotopy invariants learned in algebraic topology. Now, by using our own alternative formulation of the ``identity type" which provides an explicit formal account of ``path", operationally understood as an invertible sequence of rewrites (such as Church’s ``conversion" between $\lambda$-terms), and interpreted as a homotopy, we have been engaged in a series of papers to show examples of calculating fundamental groups of surfaces such as the circle, the torus, the 2-holed torus, the Klein bottle,  the real projective plane, and a few more. We would like to suggest that these examples  bear witness to the positive impact of MLTT in mathematics by offering formal tools to calculate and prove fundamental groups.

In a different application of computational paths, we have worked on a series of publications towards a generalisation of domain theory by defining a model of type-free $\lambda$-calculus with a groupoid structure. The idea is that solving recursive domain equations over a Cartesian closed $0$-category is a way to find extensional models of the $\lambda$-calculus. In \cite{Martinez-Rivillas-phd} and a series of papers \cite{Martinez-Rivillas-igpl,Martinez-Rivillas-aml,Martinez-Rivillas-bsl} we seek to generalise these equations to ``homotopy domain equations'' in order to build a particular Cartesian closed ``$(0,\infty)$-category'', which we call the Kleisli $\infty$-category, and thus finding higher $\lambda$-models, which are referredto as  ``$\lambda$-homotopic models''. To arrive at those objectives, we had to previously generalise c.p.o's (complete partial orders) to c.h.p.o's (complete  homotopy partial orders); complete ordered sets to complete (weakly) ordered Kan complexes, $0$-categories to $(0,\infty)$-categories and the Kleisli bicategory to a Kleisli $\infty$-category. Syntactical $\lambda$-models (e.g., the set $D_\infty$), defined on sets, are generalised to ``homotopic syntactical $\lambda$-models'' (e.g., the Kan complex ``$K_\infty$''), which are defined on Kan complexes, and go further to study the relationship of these models with the homotopic $\lambda$-model. Additionally, from the syntactic point of view, we explore what the theory of an arbitrary homotopic $\lambda$-model would be like, which turns out to contain a theory of higher $\lambda$-calculus, which we call Homotopy Type-Free Theory (HoTFT); with higher $\beta\eta$-contractions and thus with higher $\beta\eta$-conversions.

As for the basis for the formulation of the logical system which we call labelled natural deduction, we have developed a series of justified connections between a wide spectrum of concepts and techniques ranging from the philosophical works of Wittgenstein to a fundamental collection of techniques to formalise proofs in mathematics. This has been a sustained effort to fulfill a task going back at least 35 years, and the expectation is to continue with the same determination  as always. A recent paper belonging to a series of articles which began in 1987--1988, followed by some more in the 1990s has just been published \cite{Ruy-SATS}. This view of proofs and meaning originated in \cite{Ruy-Dialectica1,Ruy-Dialectica2,Ruy-ZML,Ruy-ZML2,Ruy-PhD,Ruy-Dialectica3} and led to reformulating intuitionistic type theory which gave rise to technical results such as those documented in a book \cite{Aruy33} as well as in several articles since then up until 2023, including the present one.


\section{Computational paths} \label{cpaths}

Let us begin by introducing the main work tool, an entity known as computational paths. The proposed by \cite{Awodney2007} establishes a connection between homotopy theory and computational logic. From this, in \cite{Art3}, provides a semantic interpretation of the identity type which states that terms of the identity type can be understood as homotopic paths between two points in a space. Thus, inspired by the path-based approach of the homotopy interpretation, we can use a similar approach to define the identity type in type theory, this time as part of the deductive calculus, with a view to formalize sequences of rewrites from term to term, which we are referring to as computational paths. 

The interpretation will of course be akin to the homotopy case: a term $p : Id_{A}(a,b)$ will be a computational path between terms $a, b : A$, and such path will be the result of a sequence of rewrites. In the sequel, we shall give a formal definition. The main idea, i.e.\ proofs of equality statements as (reversible) sequences of rewrites, is not new, as can be seen in the work developed by \cite{mesenguer} that also establishes equality via the rewrite system. However, this work is based on the rewriting equity system proposed in the paper entitled "Equality in labeled deductive systems and the functional interpretation of propositional equality, presented in December 1993 at the {\em 9th Amsterdam Colloquium\/}, and published by \cite{Ruy4}.

Indeed, one of the most interesting aspects of the identity type is the fact that it can be used to construct higher structures. This is a rather natural consequence of the fact that it is possible to construct higher identities. For any $a, b : A$, we have type $Id_{A}(a,b)$. If this type is inhabited by any $p, q:Id_{A}(a,b)$, then we have type $Id_{Id_{A}(a,b)}(p,q)$. If the latter type is inhabited, we have a higher equality between $p$ and $q$ \cite{harper1}. This concept is also present in computational paths. One can show the equality between two computational paths $s$ and $t$ by constructing a third one between $s$ and $t$. We show in this chapter a system of rules used to establish equalities between computational paths \cite{Anjo1}. 

Another important question we want to answer is one that arises naturally when talking about equality: Is there a canonical proof for an expression $t_{1} = t_{2}$? In the language of computational paths, is there a normal path between $t_{1}$ and $t_{2}$ such that every other path can be reduced to this one? In \cite{Arttese}, it was proved that the answer is negative, this model also refutes the Uniqueness of Identity Proofs.

\subsection{Introducing computational paths} \label{path}

Before we get into the details of what those computational paths are, let us recall what motivated the introduction of computational paths to type theory. In type theory, the rules for the construction of elements of our types can arguably be taken to have originated the so-called Curry-Howard correspondence, namely, the Brouwer-Heyting-Kolmogorov Interpretation (BHK), where propositions are defined by what constitutes a proof of it. That way, a semantic interpretation of formulas are not given by truth-values, but by the concept of proof as a primitive notion. In our complement to the so-called BHK interpretation, we have argued for the need to define propositions not just by their proof-conditions, i.e.\ their `introduction' rules, but also by what can (immediately) be drawn from them, which amounts to the need for the definition of the so-called `reduction' rules, as defined originally by Prawitz on the basis of the so-called Inversion Principle: how do the elimination rules act in the result of the corresponding introduction rules. And this means that in our formulation of the rules of proof in Labelled Natural Deduction, the idea is to give a formal account to a (suitably defined) `harmony' between the rules on the terms (which operate on proof constructions) and the rules of deduction (which operate on propositions). In this setting, we are not committing ourselves to the completeness (or otherwise) of the BHK interpretation to intuitionistic logic.\footnote{Here we wish to thank an anonymous reviewer for having raised the issue of using BHK as ``a 
failure-free way of providing meaning to Intuitionism''. We are certainly aware of the work by Piecha, 
Schroeder-Heister and Sanz \cite{Piecha} in the context of proof-theoretic semantics, which is based on a specific reading of Wittgenstein's `meaning is use' paradigm, but here we are taking a different perspective on proofs and meaning which advocates for the meaning of a proposition being given by the explanation of its (immediate) consequences, which is formalised by the reduction rules \cite{Ruy-ZML,Ruy-ZML2,Ruy-PhD,Ruy-SATS}. This finds a parallel in game/dialogical accounts of meaning such as Hintikka's game-theoretic semantics and Lorenzen's dialogue games, with roots in Peirce's account of the `Utterer vs Interpreter' dichotomy in explaining meaning in language. In this sense, our reference to BHK is not directly related to giving meaning to Intuitionism, but rather to point to the origins of a formal system which combines a functional (lambda) calculus on the terms which represent proofs and a logical calculus (natural deduction) on the formulas, where meaning is determined by game/dialogue-like semantics. As a result, as pointed out in \cite{Ruy-SATS}, what we have here may be understood as showing that constructivist semantics need not and should not be `verificationist’, in (at least) Dummett--Prawitz--Martin-L\"of’s sense. The point to be retained from (e.g.) Brouwer’s intuitionism is its constructivism. The `solipsism’ is spurious; but, in this sense, Brouwer and Wittgenstein are correct that understanding must be 1st person, though anyone’s achieving such understanding can of course be within social and worldly contexts, including chalk boards.}
Thus, we have \cite{Ruy1}:

\begin{tabbing}
	\textbf{a proof of the proposition:} \hbox{\ \ \ \ \ } \= \textbf{is given by:} \\
	$A\land B$ \> a proof of $A$ \textbf{and} a proof of $B$ \\ 
	$A\lor B$ \> a proof of $A$ \textbf{or} a proof of $B$ \\
	$A\rightarrow B$ \> a \textbf{function} that turns a proof of $A$ \\
 \> into a proof of $B$ \\
	$\forall x^D.P(x)$ \> a \textbf{function} that turns an element $a$\\
 \> into a proof of $P(a)$ \\
	$\exists x^D.P(x)$ \> an element $a$ (witness) \textbf{and a proof of} $P(a)$ \\
\end{tabbing}

Also, based on the Curry-Howard functional interpretation of logical connectives, we have \cite{Ruy1}:

\begin{tabbing}
	\textbf{a proof of the proposition:} \hbox{\ \ \ \ \ } \= \textbf{has the canonical form of:} \\
	$A\land B$ \> $\langle p,q\rangle$ where $p$ is a proof of $A$ and \\
 \> $q$ is a proof of $B$ \\
	$A\lor B$ \> $i(p)$ where $p$ is a proof of $A$ or\\
 \> $j(q)$ where $q$ is a proof of $B$ \\
	\> (`$i$' and `$j$' abbreviate `into the left/right\\ 
 \> disjunct') \\
	$A\rightarrow B$ \> $\lambda x.b(x)$ where $b(p)$ is a proof of B \\
	\> provided $p$ is a proof of A \\
	$\forall x^A.B(x)$ \> $\Lambda x.f(x)$ where $f(a)$ is a proof of $B(a)$ \\ 
	\> provided $a$ is an arbitrary individual chosen\\
	\> from the domain $A$\\ 
	$\exists x^A.B(x)$ \> $\varepsilon x.(f(x),a)$ where $a$ is a witness\\
	\> from the domain $A$, $f(a)$ is a proof of $B(a)$ \\
\end{tabbing}

If one looks closely, there is one interpretation missing in the BHK-Interpretation. What constitutes a proof of $t_{1} = t_{2}$? In \cite{Ruy1} it was proposed that an equality between these two terms should be a sequence of rewritings starting at $t_{1}$ and ending at $t_{2}$

We answer this by proposing that an equality between those two terms should be a sequence of rewrites starting from $t_{1}$ and ending at $t_{2}$. Thus, we would have \cite{Ruy1}:

\begin{tabbing}
	\textbf{a proof of the proposition:} \hbox{\ \ \ \ \ } \= \textbf{is given by:} \\
	\\
	$t_1= t_2$ \> ? \\
	\> (Perhaps a sequence of rewrites \\
	\> starting from $t_1$ and ending in $t_2$?) \\
\end{tabbing}

We call computational path the sequence of rewrites between these terms.

\subsection{Formal definition}


Before we define formally a computational path, we can take a look at one famous equality theory, the $\lambda\beta\eta-equality$ \cite{lambda}:

\begin{definition}
	The \emph{$\lambda\beta\eta$-equality} is composed by the following axioms:
	
	\begin{enumerate}
		\item[$(\alpha)$] $\lambda x.M = \lambda y.M[y/x]$ \quad if $y \notin FV(M)$;
		\item[$(\beta)$] $(\lambda x.M)N = M[N/x]$;
		\item[$(\rho)$] $M = M$;
		\item[$(\eta)$] $(\lambda x.Mx) = M$ \quad $(x \notin FV(M))$.
	\end{enumerate}
	
	And the following rules of inference:

	\bigskip
	\noindent
	\begin{bprooftree}
		\AxiomC{$M = M'$ }
		\LeftLabel{$(\mu)$ \quad}
		\UnaryInfC{$NM = NM'$}
	\end{bprooftree}
	\begin{bprooftree}
            \hskip 40pt
		\AxiomC{$M = N$}
		\AxiomC{$N = P$}
		\LeftLabel{$(\tau)$}
		\BinaryInfC{$M = P$}
	\end{bprooftree}
	
	\bigskip
	\noindent
	\begin{bprooftree}
		\AxiomC{$M = M'$ }
		\LeftLabel{$(\nu)$ \quad}
		\UnaryInfC{$MN = M'N$}
	\end{bprooftree}
	\begin{bprooftree}
            \hskip 40pt
		\AxiomC{$M = N$}
		\LeftLabel{$(\sigma)$}
		\UnaryInfC{$N = M$}
	\end{bprooftree}
	
	\bigskip
	\noindent
	\begin{bprooftree}
		\AxiomC{$M = M'$ }
		\LeftLabel{$(\xi)$ \quad}
		\UnaryInfC{$\lambda x.M= \lambda x.M'$}
	\end{bprooftree}
	\begin{bprooftree}
        \hskip 32pt
	\AxiomC{$Mx = Nx$ }
	\LeftLabel{$(\zeta)$ \quad}
	\UnaryInfC{$M = N$}
	\end{bprooftree}
	\medskip
	
	If $M = N$ is provable in $\lambda\beta\eta$, then we say that  $\lambda\beta\eta \vdash M = N$.
	
\end{definition}

\begin{definition}[\cite{lambda}]
	$P$ is $\beta$-equal or $\beta$-convertible to $Q$  (notation $P=_\beta Q$)
	iff $Q$ is obtained from $P$ by a finite (perhaps empty)  series of $\beta$-contractions
	and reversed $\beta$-contractions  and changes of bound variables.  That is,
	$P=_\beta Q$ iff \textbf{there exist} $P_0, \ldots, P_n$ ($n\geq 0$)  such that
	$P_0\equiv P$,  $P_n\equiv Q$,
	$(\forall i\leq n-1) (P_i\triangleright_{1\beta}P_{i+1}  \mbox{ or }P_{i+1}\triangleright_{1\beta}P_i  \mbox{ or } P_i\equiv_\alpha P_{i+1}).$
\end{definition}
\noindent (Note that equality has an \textbf{existential} force, which will show in the proof rules for the identity type.)

The same happens with $\lambda\beta\eta$-equality:
\begin{definition}[$\lambda\beta\eta$-equality \citep{lambda}]
	The equality-relation determined by the theory $\lambda\beta\eta$ is called
	$=_{\beta\eta}$; that is, we define
	\[M=_{\beta\eta}N\quad\Leftrightarrow\quad\lambda\beta\eta\vdash M=N.\]
\end{definition}

\begin{example}
	Take the term $M\equiv(\lambda x.(\lambda y.yx)(\lambda w.zw))v$. Then, it is $\beta\eta$-equal to $N\equiv zv$ because of the sequence:\\
	$(\lambda x.(\lambda y.yx)(\lambda w.zw))v, \quad  (\lambda x.(\lambda y.yx)z)v, \quad   (\lambda y.yv)z , \quad zv$\\
	which starts from $M$ and ends with $N$, and each member of the sequence is obtained via 1-step $\beta$- or $\eta$-contraction of a previous term in the sequence. In \cite{Art3} we can find the example that shows that this sequence can be This sequence can be seen as a path.
\end{example}

\begin{example}\label{examplepath}
	The term $M\equiv(\lambda x.(\lambda y.yx)(\lambda w.zw))v$ is $\beta\eta$-equal to $N\equiv zv$ because of the sequence:\\
	$(\lambda x.(\lambda y.yx)(\lambda w.zw))v, \quad  (\lambda x.(\lambda y.yx)z)v, \quad   (\lambda y.yv)z , \quad zv$\\
	Now, taking this sequence into a path leads us to the following:\\

 The first is equal to the second based on the grounds:\\
	$\eta((\lambda x.(\lambda y.yx)(\lambda w.zw))v,(\lambda x.(\lambda y.yx)z)v)$\\
	The second is equal to the third based on the grounds:\\
	$\beta((\lambda x.(\lambda y.yx)z)v,(\lambda y.yv)z)$\\
	Now, the first is equal to the third based on the grounds:\\
	$\tau(\eta((\lambda x.(\lambda y.yx)(\lambda w.zw))v,(\lambda x.(\lambda y.yx)z)v),\beta((\lambda x.(\lambda y.yx)z)v,(\lambda y.yv)z))$\\
	Now, the third is equal to the fourth one based on the grounds:\\
	$\beta((\lambda y.yv)z,zv)$\\
	Thus, the first one is equal to the fourth one based on the grounds:\\
	$\tau(\tau(\eta((\lambda x.(\lambda y.yx)(\lambda w.zw))v,(\lambda x.(\lambda y.yx)z)v),\beta((\lambda x.(\lambda y.yx)z)v,(\lambda y.yv)z)),\beta((\lambda y.yv)z,zv)))$.
\end{example}

In \cite{Art3} we can find in example $2.4$ how a sequence can be constructed that establishes a $\lambda\beta\eta$-equality between two terms, and in example $3.1$ how this sequence can be seen as one way. 

The aforementioned theory establishes the equality between two $\lambda$-terms. Therefore, two terms are $\lambda\beta\eta$-equal if {\bf there exists} a sequence of applications of rules of definitional equality which builds
a proof of their equality in the theory of $\lambda\beta\eta$-equality. Since we are working with computational objects as terms of a type, we can consider the following definition:

\begin{definition}
	The equality theory of Martin L\"of's type theory has the following basic proof rules for the $\Pi$-type:
	
	\bigskip
	
	\noindent
	\begin{bprooftree}
		\hskip -2pt
		\alwaysNoLine
		\AxiomC{$N : A$}
		\AxiomC{$[x : A]$}
		\UnaryInfC{$M : B$}
		\alwaysSingleLine
		\LeftLabel{$(\beta$) \quad}
		\BinaryInfC{$(\lambda x.M)N = M[N/x] : B[N/x]$}
	\end{bprooftree}
	\begin{bprooftree}
		\hskip 8pt
		\alwaysNoLine
		\AxiomC{$[x : A]$}
		\UnaryInfC{$M = M' : B$}
		\alwaysSingleLine
		\LeftLabel{$(\xi)$ \quad}
		\UnaryInfC{$\lambda x.M = \lambda x.M' : (\Pi x : A)B$}
	\end{bprooftree}
	
	\bigskip
	
	\noindent
	\begin{bprooftree}
		\hskip -2pt
		\AxiomC{$M : A$}
		\LeftLabel{$(\rho)$ \quad}
		\UnaryInfC{$M = M : A$}
	\end{bprooftree}
	\begin{bprooftree}
		\hskip 90pt
		\AxiomC{$M = M' : A$}
		\AxiomC{$N : (\Pi x : A)B$}
		\LeftLabel{$(\mu)$ \quad}
		\BinaryInfC{$NM = NM' : B[M/x]$}
	\end{bprooftree}
	
	\bigskip
	
	\noindent
	\begin{bprooftree}
		\hskip -2pt
		\AxiomC{$M = N : A$}
		\LeftLabel{$(\sigma) \quad$}
		\UnaryInfC{$N = M : A$}
	\end{bprooftree}
	\begin{bprooftree}
		\hskip 74pt
		\AxiomC{$N : A$}
		\AxiomC{$M = M' : (\Pi x : A)B$}
		\LeftLabel{$(\nu)$ \quad}
		\BinaryInfC{$MN = M'N : B[N/x]$}
	\end{bprooftree}
	
	\bigskip
	
	\noindent
	\begin{bprooftree}
		\hskip -2pt
		\AxiomC{$M = N : A$}
		\AxiomC{$N = P : A$}
		\LeftLabel{$(\tau)$ \quad}
		\BinaryInfC{$M = P : A$}
	\end{bprooftree}
	\begin{bprooftree}
		\hskip 1pt
		\AxiomC{$M: (\Pi x : A)B$}
		\LeftLabel{$(\eta)$ \quad}
		\UnaryInfC{$(\lambda x.Mx) = M: (\Pi x : A)B$}
	\end{bprooftree}
	
	\bigskip
	
\end{definition}

Where $x \notin FV(M)$ in $\eta-rule$. Now, we are finally able to formally define computational paths:

\begin{definition}
	Let $a$ and $b$ be elements of a type $A$. We say that $s$ is an \emph {computational path} from $ a $ to $ b $, if $s$ is a composite of a sequence of rewrites (each rewrite is an application of the inference rules of the equality theory of types theory or is a change of bounded variables) that when applied to the term $a$, we get the term $b$ . We denote that by $a =_{s} b$.
\end{definition}

As we have seen in \emph{example \ref{examplepath}}, composition of rewrites are applications of the rule $\tau$. Since change of bound variables is possible, each term is considered up to $\alpha$-equivalence.

\subsection{Equality equations}

Using the axioms of \emph{$\lambda\beta\eta$-equality} and the equality theory of Martin L\"of's type theory, we can show that computational paths establishes the three fundamental equations of equality: The transitivity (\textbf{$\tau$}),  Reflexivity  (\textbf{$\rho$}) and Symmetry (\textbf{$\sigma$}) which can be seen below respectively.

\bigskip

\begin{bprooftree}
\AxiomC{$a =_{t} b : A$}
\AxiomC{$b =_{u} c : A$}
\RightLabel{\large\textit{\textbf{($\tau$)}}}
\BinaryInfC{$a =_{\tau(t,u)} c : A$}
\end{bprooftree}
\begin{bprooftree}
\AxiomC{$a : A$}
\RightLabel{\textit\large\textit{\textbf{($\rho$)}}}
\UnaryInfC{$a =_{\rho} a : A$}
\end{bprooftree}
\begin{bprooftree}
\AxiomC{$a =_{t} b : A$}
\RightLabel{\textit\large\textit{\textbf{($\sigma$)}}}
\UnaryInfC{$b =_{\sigma(t)} a : A$}
\end{bprooftree}

\bigskip
\begin{bprooftree}
\AxiomC{$a =_{t} b : A$}
\RightLabel{\textit{s}}
\UnaryInfC{$b =_{\sigma(t)} a : A$}
\end{bprooftree}

\subsection{Identity type}

We have said that it is possible to formulate the identity type using computational paths. As we have seen, the best way to define any formal entity of type theory is by a set of natural deductions rules. Thus, we define our path-based approach as the following set of rules: 

\begin{itemize}
	
\item Formation and Introduction rules \cite{Ruy1}, \cite{Art3}:
	\begin{center}
		\begin{bprooftree}
			\AxiomC{$A$ type}
			\AxiomC{$a : A$}
			\AxiomC{$b : A$}
			\RightLabel{$Id - F$}
			\TrinaryInfC{$Id_{A}(a,b)$ type}
		\end{bprooftree}
  \begin{bprooftree}
			\AxiomC{$a =_{s} b : A$}
			\RightLabel{$Id - I$}
			\UnaryInfC{$s(a,b) : Id_{A}(a,b)$}
		\end{bprooftree}
	
	\bigskip
	
		\begin{bprooftree}
			\AxiomC{$a =_{s} b : A$}
			\RightLabel{$Id - I$}
			\UnaryInfC{$s(a,b) : Id_{A}(a,b)$}
		\end{bprooftree}
	\end{center}
	
	One can notice that our formation rule is exactly equal to the traditional identity type. From terms $a, b : A$, one can form that is inhabited only if there is a proof of equality between those terms, i.e., $Id_{A}(a,b)$.
	
	The difference starts with the introduction rule. In our approach, one can notice that we do not use a reflexive constructor $r$. In other words, the reflexive path is not the main building block of our identity type. Instead, if we have a computational path $a =_{s} b : A$, we introduce $s(a,b)$ as a term of the identity type. That way, one should see $s(a,b)$ as a sequence of rewrites and substitutions (i.e., a computational path) which would have started from $a$ and arrived at $b$
	
	\medskip 
	
	\item Elimination rule \cite{Ruy1}, \cite{Art3}:
	
	\begin{center}
		\begin{bprooftree}
			\alwaysNoLine
			\AxiomC{$m : Id_{A}(a,b)$ }
			\AxiomC{$[a =_{g} b : A]$}
			\UnaryInfC{$h(g) : C$}
			\alwaysSingleLine
			\RightLabel{$Id - E$}
			\BinaryInfC{$REWR(m, \acute{g}.h(g)) : C$}
		\end{bprooftree}
	\end{center}
	
	Let's recall the notation being used. First, one should see $h(g)$ as a functional expression $h$ which depends on $g$. Also, one should notice the use of `$\acute{\ }$' in $\acute{g}$. One should see `$\acute{\ }$' as an abstractor that binds the occurrences of the variable $g$ introduced in the local assumption $[a =_{g} b : A]$ as a kind of {\em Skolem-type\/} constant denoting the {\em reason\/} why $a$ was assumed to be equal to $b$.
	
	We also introduce the constructor $REWR$. In a sense, it is similar to the constructor $J$ of the original elimination operator of MLTT, since both arise from the elimination rule of the identity type. The behavior of $REWR$ is simple. If from a computational path $g$ that establishes the equality between $a$ and $b$ one can construct $h(g) : C$, then if we also have this equality established by a term $C$, we can put together all this information in $REWR$ to construct $C$, eliminating the type $Id_{A}(a,b)$ in the process. The idea is that we can substitute $g$ for $m$ in $\acute{g}.h(g)$, resulting in $h(m/g) : C$. This behavior is established next by the reduction rule.\\
	
	\item Reduction rule 	\cite{Ruy1}, \cite{Art3}:
	
	\begin{center}
		\begin{bprooftree}
			\AxiomC{$a =_{m} b : A$}
			\RightLabel{$Id - I$}
			\UnaryInfC{$m(a,b) : Id_{A}(a,b)$}
			\alwaysNoLine
			\AxiomC{$[a =_{g} b : A]$}
			\UnaryInfC{$h(g) : C$}
			\alwaysSingleLine
			\RightLabel{$Id - E$ \quad $\rhd_\beta$}
			\BinaryInfC{$REWR(m, \acute{g}.h(g)) : C$}
		\end{bprooftree}
		\begin{bprooftree}
			\AxiomC{$[a =_{m} b : A]$}
			\alwaysNoLine
			\UnaryInfC{$h(m/g):C$}
		\end{bprooftree}
	\end{center}
	
	\item Induction rule: 
	
	\begin{center}
		\begin{bprooftree}
			\AxiomC{$e : Id_{A}(a,b)$}
			\AxiomC{$[a =_{t} b : A]$}
			\RightLabel{$Id - I$}
			\UnaryInfC{$t(a, b) : Id_{A}(a, b)$}
			\RightLabel{$Id - E$ \quad  $\rhd_{\eta}$ \quad $e : Id_{A}(a,b)$}
			\BinaryInfC{$REWR(e, \acute{t}.t(a,b)) : Id_{A}(a,b)$}
		\end{bprooftree}
	\end{center}
	
\end{itemize}

Our introduction and elimination rules reassure the concept of equality as an \textbf{existential force}. In the introduction rule, we encapsulate the idea that a witness of a identity type $Id_{A}(a,b)$ only exists if there exist a computational path establishing the equality of $a$ and $b$. Also, one can notice that elimination rule is similar to the elimination rule of the existential quantifier.

\subsection{Path-based examples}

The objective of this subsection is to show how to use in practice the rules that we have just defined. The idea is to show construction of terms of some important types. The constructions that we have chosen to build are the reflexive, transitive and symmetric type of the identity type. Those were not random choices. The main reason is the fact that reflexive, transitive and symmetric types are essential to the process of building a groupoid model for the identity type \cite{hofmann1}. As we shall see, these constructions come naturally from simple computational paths constructed by the application of axioms of the equality of type theory.

In the constructions of terms of some important types, we chose to construct the reflexive, transitive and symmetric types of the identity type. These were not random choices. The main reason is the fact that reflexive, transitive and symmetric types are essential for the process of building a groupoid model for the identity type \cite{hofmann1}. 

These constructions come naturally from simple computational paths constructed by the application of axioms of the equality of type theory.  The process of building a term of some type is a matter of finding the right reason. In the case of $J$, the reason is the correct $x,y : A$ and $z : Id_{A}(a,b)$ that generates the adequate $C(x,y,z)$. In our approach, the reason is the correct path $a =_{g} b$ that generates the adequate $g(a,b) : Id(a,b)$. A proof of how to construct reflexivity, transitivity and symmetry can be found subsection 3.1 in \cite{Art3}.

One could find strange the fact that we need to prove the reflexivity. Nevertheless, just remember that our approach is not based on the idea that reflexivity is the base of the identity type. As usual in type theory, a proof of something comes down to a construction of a term of a type. In this case, we need to construct a term of type $\Pi_{(a : A)}Id_{A}(a,a)$. The reason is extremely simple: from a term $a : A$, we obtain the computational path $a =_{\rho} a : A$ \cite{Art3}:

\begin{center}
	\begin{bprooftree}
		\AxiomC{$[a : A]$}
		\UnaryInfC{$a =_{\rho} a : A$}
		\RightLabel{$Id - I$}
		\UnaryInfC{$\rho(a,a) : Id_{A}(a,a)$}
		\RightLabel{$\Pi-I$}
		\UnaryInfC{$\lambda a.\rho(a,a) : \Pi_{(a : A)}Id_{A}(a,a)$}
	\end{bprooftree}
\end{center}

\subsubsection{Symmetry}

The second proposed construction is the symmetry. Our objective is to obtain a term of type $\Pi_{(a : A)}\Pi_{(b : A)}(Id_{A}(a,b)\rightarrow Id_{A}(b,a))$.

We construct a proof using computational paths. As expected, we need to find a suitable reason. Starting from $a =_{t} b$, we could look at the axioms of \emph{definition 4.1} to plan our next step. One of those axioms makes the symmetry clear: the $\sigma$ axiom. If we apply $\sigma$, we will obtain $b =_{\sigma(t)} a$. From this, we can then infer that $Id_A$ is inhabited by $(\sigma(t))(b,a)$. Now, it is just a matter of applying the elimination \cite{Art3}:

\begin{center}
        \resizebox{0.9\hsize}{!}{	\begin{bprooftree}
		\alwaysNoLine
		\AxiomC{$[a:A] \quad [b:A]$}
		\UnaryInfC{$[p(a,b) : Id_{A}(a,b)]$}
		\alwaysSingleLine
		\AxiomC{[$a =_{t} b : A$]}
		\UnaryInfC{$b =_{\sigma(t)} a : A$}
		\RightLabel{$Id - I$}
		\UnaryInfC{$(\sigma(t))(b,a) : Id_{A}(b,a)$}
		\RightLabel{$Id - E$}
		\BinaryInfC{$REWR(p(a,b),\acute{t}.(\sigma(t))(b,a)) : Id_{A}(b,a)$}
		\RightLabel{$\rightarrow - I$}
		\UnaryInfC{$\lambda p.REWR(p(a,b), \acute{t}.(\sigma(t))(b,a)) : Id_{A} (a,b) \rightarrow Id_{A}(b,a)$}
		\RightLabel{$\Pi-I$}
	\UnaryInfC{$\lambda b. \lambda p.REWR(p(a,b),\acute{t}.(\sigma(t))(b,a)) :  \Pi_{(b : A)}(Id_{A} (a,b) \rightarrow Id_{A}(b,a))$}
	\RightLabel{$\Pi-I$}
		\UnaryInfC{$\lambda a.\lambda b. \lambda p.REWR(p(a,b), \acute{t}.(\sigma(t))(b,a)) :  \Pi_{(a : A)}\Pi_{(b : A)}(Id_{A} (a,b) \rightarrow Id_{A}(b,a))$}
	\end{bprooftree}}
\end{center}

\subsubsection{Transitivity}
The third and last construction will be the transitivity. Our objective is to obtain a term of type 
 $\Pi_{(a : A)}\Pi_{(b : A)}\Pi_{(c : A)} (Id_{A}(a,b) \rightarrow Id_{A}(b,c) \rightarrow Id_{A}(a,c))$.

To build our path-based construction, the first step, as expected, is to find the reason. Since we are trying to construct the transitivity, it is natural to think that we should start with paths $a =_{t} b$ and $b =_{u} c$ and then, from these paths, we should conclude that there is a path $z$ that establishes that $a =_{z} c$. To obtain $z$, we could try to apply the axioms of \emph{definition 4.1}. Looking at the axioms, one is exactly what we want: the axiom $\tau$. If we apply $\tau$ to  $a =_{t} b$ and $b =_{u} c$, we will obtain a new path $\tau(t,u)$ such that $a = _{\tau(t,u)} c$. Using that construction as the reason, we obtain the following term \cite{Art3}:

\begin{center}
 \resizebox{1.1\hsize}{!}{
		\begin{bprooftree}
			\alwaysNoLine
			\AxiomC{$[a:A] \quad [b:A]$}
			\UnaryInfC{$[w(a,b) : Id_{A}(a,b)]$}
			\alwaysNoLine
			\AxiomC{$[c:A]$}
			\UnaryInfC{$[s(b,c) : Id_{A}(b,c)]$}
			\alwaysSingleLine
			\AxiomC{$[a =_{t} b:A]$}
			\AxiomC{$[b =_{u} c:A]$}
			\BinaryInfC{$a =_{\tau(t,u)} c:A$}
			\RightLabel{$Id - I$}
			\UnaryInfC{$(\tau (t,u))(a,c) : Id_{A}(a,c)$}
			\RightLabel{$Id - E$}
			\BinaryInfC{$REWR(s(b,c),\acute{u}(\tau (t,u))(a,c)) : Id_{A}(a,c)$}
			\RightLabel{$Id - E$}
			\BinaryInfC{$REWR(w(a,b),\acute{t}REWR(s(b,c),\acute{u}(\tau (t,u))(a,c))) : Id_{A}(a,c)$}
			\RightLabel{$\rightarrow - I$}
			\UnaryInfC{$\lambda s.REWR(w(a,b),\acute{t}REWR(s(b,c),\acute{u}(\tau (t,u))(a,c))) : Id_{A}(b,c) \rightarrow Id_{A}(a,c)$}
			\RightLabel{$\rightarrow - I$}
			\UnaryInfC{$\lambda w.\lambda s.REWR(w(a,b),\acute{t}REWR(s(b,c),\acute{u}(\tau (t,u))(a,c))) : Id_{A}(a,b) \rightarrow Id_{A}(b,c) \rightarrow Id_{A}(a,c)$}
			\RightLabel{$\Pi-I$}
			\UnaryInfC{$\lambda c.\lambda w.\lambda s.REWR(w(a,b),\acute{t}REWR(s(b,c),\acute{u}(\tau (t,u))(a,c))) :  \Pi_{(c : A)}(Id_{A}(a,b) \rightarrow Id_{A}(b,c) \rightarrow Id_{A}(a,c))$}
			\RightLabel{$\Pi-I$}
			\UnaryInfC{$\lambda b. \lambda c.\lambda w.\lambda s.REWR(w(a,b),\acute{t}REWR(s(b,c),\acute{u}(\tau (t,u))(a,c))) :  \Pi_{(b : A)}\Pi_{(c : A)}(Id_{A}(a,b) \rightarrow Id_{A}(b,c) \rightarrow Id_{A}(a,c))$}
			\RightLabel{$\Pi-I$}
			\UnaryInfC{$\lambda a. \lambda b. \lambda c.\lambda w.\lambda s.REWR(w(a,b),\acute{t}REWR(s(b,c),\acute{u}(\tau (t,u))(a,c))) :   \Pi_{(a : A)}\Pi_{(b : A)}\Pi_{(c : A)}(Id_{A}(a,b) \rightarrow Id_{A}(b,c) \rightarrow Id_{A}(a,c))$}
		\end{bprooftree}}
	\end{center}
\medskip
As one can see, each step is just straightforward applications of introduction, elimination rules and abstractions. The only idea behind this construction is just the simple fact that the axiom $\tau$ guarantees the transitivity of paths.

\subsection{Term rewriting system}

As we have just shown, a computational path establishes when two terms of the same type are equal. From the theory of computational paths, an interesting case arises. Suppose we have a path $s$ that establishes that $a =_{s} b : A$ and a path $t$ that establishes that $a =_{t} b : A$. Consider that $s$ and $t$ are formed by distinct compositions of rewrites. Is it possible to conclude that there are cases that $s$ and $t$ should be considered equivalent? The answer is \emph{yes}. Consider the following example \cite{Arttese}:



	

\begin{example}
	\noindent \normalfont Consider a path $a =_{t} b : A$. Applying the symmetry, one ends up with $b =_{\sigma(t)} a : A$. One can take those two paths and apply the transitivity, ending up with $a =_{\tau(t,\sigma(t))} a$. Since the path $\tau$ is the inverse of the $\sigma(\tau)$, the composition of those two paths should be equivalent to the reflexive path. Thus, $\tau(t,\sigma(t))$ should be reduced to $\rho$.
\end{example}	
	

The example is a simple and straightforward case and shows that different paths should be considered equal if one is just a redundant form of the other.
Since the equality theory has a total of 7 axioms, the possibility of combinations that could generate redundancies are high. Fortunately, all possible redundancies have been thoroughly mapped out by \cite{Anjo1} and \cite{Ruy1}. There they established a system known as $LND_{EQ}-TRS$, which establishes all redundancies and creates rules that resolve them, totaling $39$ rules. All this $39$ rules are arranged in the appendix.

 When we talk about these kind of systems, two questions arise: Every computational path has a normal form? And if a computational path has a normal form, is it unique? To show that it has a normal form, one has to prove that every computational path terminates, i.e., after a finite number of rewrites, one will end up with a path that does not have any additional reduction. To show that it is unique, one needs to show that the system is confluent. In other words, if one has a path with $2$ or more reductions, one needs to show that the choice of the rewrite rule does not matter. In the end, one will always obtain the same end-path without any redundancies. Full proof of these requirements can be found at \cite{Anjo1},\cite{Ruy2},\cite{Ruy3},\cite{RuyAnjolinaLivro} and \cite{dershowitz}.

\subsection{Rewrite equality}

From the $39$  $LND_{EQ}-TRS$ rules, we have the following definition:

\begin{definition}[Rewrite Rule \citep{Art3}]
 An $rw$-rule is any of the rules defined in $LND_{EQ}-TRS$.
\end{definition}

Similarly to the $\beta$-reduction of $\lambda$-calculus, we have a definition for rewrite reduction:

\begin{definition}[Rewrite reduction \citep{Art3}]
	Let $s$ and $t$ be computational paths. We say that $s \rhd_{1rw} t$ (read as: $s$ $rw$-contracts to $t$) iff we can obtain $t$ from $s$ by an application of only one $rw$-rule. If $s$ can be reduced to $t$ by finite number of $rw$-contractions, then we say that $s \rhd_{rw} t$ (read as $s$ $rw$-reduces to $t$).
	
\end{definition}

We also have rewrite contractions and equality:

\begin{definition}[Rewrite contraction and equality \citep{Art3}]
Let $s$ and $t$ be computational paths. We say that $s =_{rw} t$ (read as: $s$ is $rw$-equal to $t$) iff $t$ can be obtained from $s$ by a finite (perhaps empty) series of $rw$-contractions and reversed $rw$-contractions. In other words, $s =_{rw} t$ iff there exists a sequence $R_{0},....,R_{n}$, with $n \geq 0$, such that $(\forall i \leq n - 1) (R_{i}\rhd_{1rw} R_{i+1}$ or $R_{i+1} \rhd_{1rw} R_{i})$, $R_{0} \equiv s$ and \quad $R_{n} \equiv t$.
\end{definition}

Thus, as result we have the fact that rewrite equality is an equivalence relation \cite{Art3}:
 
\begin{proposition}\label{proposition3.7} Rewrite equality is transitive, symmetric and reflexive.
\end{proposition}

\begin{proof}
	Comes directly from the fact that $rw$-equality is the transitive, reflexive and symmetric closure of $rw$.
\end{proof}

Rewrite reduction and equality play fundamental roles in the groupoid model of a type based on computational paths, as we are going to see in the sequel.

\subsection{LND$_{EQ}$-TRS(2)}

We know there are redundancies which are resolved by a system called $LND_{EQ}-TRS$.  In fact, since these axioms just define an equality theory for type theory, one can specify and say that these are redundancies of the equality of type theory. As we mentioned, the $LND_{EQ}-TRS$ has a total of $39$ rules \cite{Anjo1},\cite{Ruy1}. Since the $rw$-equality is based on the rules of  $LND_{EQ}-TRS$,  one can just imagine the high number of redundancies that $rw$-equality could cause. 

In fact, a thorough study of all the redundancies caused by these rules led to the work done in \cite{Arttese}, that only interested in the redundancies caused by the fact that $rw$-equality is transitive, reflexive and symmetric with the addition of only one specific $rw_{2}$-rule. 

This way up, was created a system, called $LND_{EQ}-TRS_{2}$, that resolves all the redundancies caused by $rw$-equality (the same way that $LND_{EQ}-TRS$ resolves all the redundancies caused by equality), creating new rewrite rules denoted by $tr_{2}$, $tsr_{2}$, $trr_{2}$, $tlr_{2}$, $sr_{2}$, $ss_{2}$ and $ tt_{2}$. Since $rw$-equality is just a sequence of $rw$-rules (also similar to equality, since equality is just a computational path, i.e., a sequence of identifiers), then we could put a name on these sequences.  The $LND_{EQ}-TRS_{2}$ system can be seen at \cite{Arttese} and the only difference is that instead of having $rw$-rules and $rw$-equality, we have $rw_{2}$-rules and $rw_{2}$-equality.

There is an important rule specific to this system. It stems from the fact that transitivity of reducible paths can be reduced in different ways, but generating the same result. For example, consider the simple case of $\tau(s,t)$ and consider that it is possible to reduce $s$ to $s'$ and $t$ to $t'$. There is two possible $rw$-sequences that reduces this case: The first one is $\theta: \tau(s,t) \rhd_{1rw} \tau(s',t) \rhd_{1rw} \tau(s',t')$ and the second $\theta': \tau(s,t) \rhd_{1rw} \tau(s,t') \rhd_{1rw} \tau(s',t')$. Both $rw$-sequences obtained the same result in similar ways, the only difference being the choices that have been made at each step. Since the variables, when considered individually,  followed the same reductions, these $rw$-sequences should be considered redundant relative to each other and, for that reason, there should be $rw_{2}$-rule that establishes this reduction. This rule is called \emph{independence of choice} and is denoted by $cd_{2}$. Since we already understand the necessity of such a rule, we can define it formally:

\begin{definition}[Independence of choice \cite{Art3}]
	Let $\theta$ and $\phi$ be $rw$-equalities expressed by two $rw$-sequences: $\theta: \theta_{1},...,\theta_{n}$, with $n \geq 1$, and $\phi: \phi_{1},...,\phi_{m}$, with $m \geq 1$. Let $T$ be the set of all possible $rw$-equalities from $\tau(\theta_{1},\phi_{1})$ to $\tau(\theta_{n},\theta_{m})$ described by the following process: $t \in T$ is of the form  $\tau(\theta_{l_{1}},\phi_{r_{1}}) \rhd_{1rw} \tau(\theta_{l_{2}},\phi_{r_{2}}) \rhd_{1rw} ... \rhd_{1rw} \tau(\theta_{l_{x}},\phi_{r_{y}})$, with $l_{1} = 1, r_{1} = 1$,  $l_{x} = n, r_{y} = m$ and $l_{i + 1} = 1 + l_{i}$ and $r_{i + 1} = r_{i}$ or $l_{i + 1} = l_{i}$ and $r_{i + 1} = 1 + r_{i}$. The independence of choice, denoted by $cd_{2}$, is defined as the rule of $LND_{EQ}-TRS_{2}$ that establishes the equality between any two different terms of $T$. In other words, if $x,y \in T$ and $x \neq y$, then $x =_{cd_{2}} y$ and $y =_{cd_{2}} x$.
	
\end{definition}

Analogously to the $rw$-equality, $rw_{2}$-equality is also an equivalence relation \cite{Art3}. The proof of this result, together with all relevant results, including the proofs of the new $7$ news rules of the $LND_{EQ}-TRS_{2}$ system, can be found in the publications \cite{Ruy1} and \cite{Art3}.

\begin{proposition}
	$rw_{2}$-equality is transitive, symmetric and reflexive.
\end{proposition}

\begin{proof}
	Analogous to Proposition \ref{proposition3.7}.
\end{proof}

\subsection{Functoriality}
We want to show that functions preserve equality \cite{hott}. We shall omit the proof of the results in this section. In these cases, where a proof has been omitted, the reference and location in the text where such proof can be checked are included.

\begin{lemma}
	The type $\Pi_{(x,y : A)}\Pi_{(f: A \rightarrow B)} (Id_{A}(x,y) \rightarrow Id_{B}(f(x),f(y)))$ is inhabited.
\end{lemma}

\begin{proof}
	It is a straightforward construction:
	
	\bigskip
	\resizebox{1.0\hsize}{!}{	\begin{bprooftree}
		\AxiomC{[$x =_{s} y : A$]}
		\AxiomC{[$f: A \rightarrow B$]}
		\BinaryInfC{$f(x) =_{\mu_{f}(s)} f(y) : B$}
		\UnaryInfC{$\mu_{f}(s)(f(x),f(y)) : Id_{B}(f(x),f(y))$}
		\AxiomC{$[p: Id_{A}(x,y)]$}
		\BinaryInfC{$REWR(p,\lambda s. \mu_{f}(s)(f(x),f(y))): Id_{B}(f(x),f(y))$}
		\UnaryInfC{$\lambda x. \lambda y. \lambda f. \lambda p.REWR(p,\lambda s. \mu_{f}(s)(f(x),f(y))):\Pi_{(x,y : A)}\Pi_{(f: A \rightarrow B)} (Id_{A}(x,y) \rightarrow Id_{B}(f(x),f(y)))$}	
	\end{bprooftree}}	
\end{proof}

\begin{lemma}\label{lemma2}
	For any functions $f: A \rightarrow B$ and $g: B \rightarrow C$ and paths $p: x =_{A} y$ and $q: y =_{A} z$, we have:
	
	\begin{enumerate}
		\item $\mu_{f}(\tau(p,q)) = \tau(\mu_{f}(p),\mu_{f}(q))$
		\item $\mu_{f}(\sigma(p)) = \sigma(\mu_{f}(p))$
		\item $\mu_{g}(\mu_{f}(p)) = \mu_{g \circ f}(p)$
		\item $\mu_{Id_{A}}(p) = p\label{rule42}$
	\end{enumerate}
\end{lemma}

\begin{proof}   The proof of this Lemma can be seen in Lemma $4.6$ of \cite{Art4}.. The Lemma 14 introduces rewriting rules $40$, $41$ and $42$.
\end{proof}

\subsection{Transport}\label{transports}

As stated in  \cite{Ruy5}, substitution can take place when no quantifier is involved. In this sense, there is a `'quantifier-less' notion of substitution. In type theory, this `'quantifier-less' substitution is given by a operation known as transport \cite{hott}. In our path-based approach, we formulate a new inference rule of `'quantifier-less' substitution \cite{Ruy5}:

\begin{prooftree}
	
	\AxiomC{$x =_{p} y : A$}
	\AxiomC{$f(x) : P(x)$}
	\BinaryInfC{$p(x,y)\circ f(x) : P(y)$}
	
\end{prooftree}

We use this transport operation to solve one essential issue of our path-based approach. We know that given a path $x =_{p} y : A$ and function $f: A \rightarrow B$, the application of axiom $\mu$ yields the path $f(x) =_{\mu_{f}(p)} f(y) : B$. The problem arises when we try to apply the same axiom for a dependent function $f : \Pi_{(x : A)} P(x)$. In that case, we want $f(x) = f(y)$, but we cannot guarantee that the type of $f(x) : P(x)$ is the same as $f(y) : P(y)$. The solution is to apply the transport operation and thus, we can guarantee that the types are the same:

\begin{center}

	\begin{prooftree}
		\AxiomC{$x =_{p} y : A$}
		\AxiomC{$f : \Pi_{(x : A)} P(x)$}
		\BinaryInfC{$p(x,y) \circ f(x) =_{\mu_{f}(p)} f(y) : P(y)$}
	\end{prooftree}
\end{center}

\begin{lemma}[Leibniz's Law]
	The type $\Pi_{(x,y: A)} (Id_{A}(x,y) \rightarrow P(x) \rightarrow P(y))$ is inhabited.
\end{lemma}

\begin{proof}
	We construct the following tree:
	\begin{prooftree}
		\AxiomC{$[x =_{p} y : A]$}
		\AxiomC{$[f(x) : P(x)]$}
		\BinaryInfC{$p(x,y)\circ f(x) : P(y)$}
		\UnaryInfC{$\lambda f(x). p(x,y) \circ f(x) : P(x) \rightarrow P(y)$}
		\AxiomC{$[z: Id_{A}(x,y)]$}
		\BinaryInfC{$REWR(z,\lambda p.\lambda f(x).p(x,y) \circ f(x)) : P(x) \rightarrow P(y)$}
		\UnaryInfC{$\lambda x. \lambda y.\lambda z.REWR(z,\lambda p.\lambda f(x).p(x,y) \circ f(x)) : \Pi_{(x,y: A)} (Id_{A}(x,y) \rightarrow P(x) \rightarrow P(y))$}
	\end{prooftree}
\end{proof}

The function $\lambda f(x). p(x,y) \circ f(x) : P(x) \rightarrow P(y)$ is usually written as $transport^{p}(p, -)$ and $transport^{p}(p,f(x)) : P(y)$ is usually written as $p_{*}(f(x))$. 

\begin{lemma}
	For any $P(x) \equiv B$, $x =_{p} y : A$ and $b : B$, there is a path $transport^{P}(p,b) = b$.
\end{lemma}

\begin{proof}
	The first thing to notice is the fact that in our formulation of transport, we always need a functional expression $f(x)$, and in this case we have only a constant term $b$. To address this problem, we consider a function $f = \lambda. b$ and then, we transport over $f(x) \equiv b$:
	
	\begin{center}
		$transport^{P}(p,f(x) \equiv b) =_{\mu(p)} (f(y) \equiv b)$. 
	\end{center}
	
	Thus, $transport^{P}(p,b) =_{\mu(p)} b$. We sometimes call this path $transportconst^{B}_{p}(b)$.
	
\end{proof}

We shall omit the proof of some results in this section. In these cases, where a proof has been omitted, the reference and location in the text where such proof can be checked are included.

\begin{lemma}\label{lemma5}
	For any $f: A \rightarrow B$ and $x =_{p} y : A$, we have 
	
	\begin{center}
		$\mu(p)(p_{*}(f(x)),f(y)) = \tau(transportconst^{B}_{p}, \mu_{f}(p))(p_{*}(f(x)),f(y))$
	\end{center}
\end{lemma}

\begin{proof} 

	The first thing to notice is that in this case, $transportconst^{B}_{p}$ is the path $\mu(p)(p*(f(x),f(x))$ by lemma \textbf{8}. As we did to the rules of $LND_{EQ}-TRS$, we establish this equality by getting to the same conclusion from the same premises by two different trees:

	In the first tree, we consider $f(x) \equiv b : B$ and transport over $b : B$:

	\begin{prooftree}
		\AxiomC{$x =_{p} y : A$}	 	
		\AxiomC{$f(x) \equiv b : B$}
		\BinaryInfC{$p(x,y) \circ (f(x) \equiv b) : B$}
	   \UnaryInfC{$p_{*}(f(x)) =_{\mu_{f}(p)} b \equiv f(x)$}
		\AxiomC{$x =_{p} y : A$}
		\AxiomC{$f: A \rightarrow B$}		
		\BinaryInfC{$f(x) =_{\mu_{f}(p)} f(y) : B$}
		\BinaryInfC{$p_{*}(f(x)) =_{\tau(\mu_{f}(p),\mu_{f}(p))} f(y) : B$}
	\end{prooftree}

	In the second one, we consider $f(x)$ as an usual functional expression and thus, we transport the usual way:
	\begin{prooftree}
		\AxiomC{$x =_{p} y : A$}
		\AxiomC{$f(x) : B$}
		\BinaryInfC{$p(x,y) \circ f(x) : B$}
		\UnaryInfC{$p_{*}(f(x)) =_{\mu_{f}(p)} f(y) : B$}
	\end{prooftree}
\end{proof}

\begin{lemma}
	For any $x =_{p} y : A$ and $q: y =_{A} z : A$, $f(x) : P(x)$, we have
	
	\begin{center}
		$q_{*}(p_{*}(f(x))) = (p\circ q)_{*}(f(x))$.
	\end{center}
\end{lemma}

\begin{proof}
	We develop both sides of the equation and wind up with the same result:
	
	\begin{center}
		$q_{*}(p_{*}f(x)) =_{\mu(p)} q_{*}(f(y)) =_{\mu(q)} f(z)$
		
		$(p\circ q)_{*}(f(x)) =_{\mu(p \circ q)} f(z)$.
	\end{center}
\end{proof}

\begin{lemma}\label{lemma7}
	For any $f: A \rightarrow B$, $x =_{p} y : A$ and $u: P(f(x))$, we have:
	
	\begin{center}
		$transport^{P \circ f}(p,u) = transport^{P}(\mu_{f}(p),u)$
	\end{center}
\end{lemma}

\begin{proof} 
	This lemma hinges on the fact that there are two possible interpretations of $u$ and this stems from the fact that $(g \circ f)(x) \equiv g(f(x))$. Thus, we can see $u$ as a functional expression $g$ on $f(x)$ or an expression $g \circ f$ on $x$:

\begin{center}
		     \resizebox{1.1\hsize}{!}{
			\begin{bprooftree}

		\AxiomC{$x =_{p} y :A$}
		\AxiomC{$ u \equiv (g \circ f)(x) : (P \circ f)(x)$}
		\BinaryInfC{$p(x,y) \circ (g \circ f)(x) : (P \circ f)(y)$}
		\UnaryInfC{$p(x,y) \circ (g \circ f)(x) =_{\mu(p)} (g \circ f)(y) : (P \circ f)(y)$}
		\UnaryInfC{$p(x,y) \circ (g \circ f)(x) =_{\mu(p)} g(f(y)) : P(f(y))$}
		
		\AxiomC{$x =_{p} y : A$}
		\UnaryInfC{$f(x) =_{u_{f}(p)} f(y) : B$}
		\AxiomC{$u \equiv g(f(x)) : P(f(x))$}
		\BinaryInfC{$\mu_{f}(p)(f(x),f(y)) \circ g(f(x)): P(f(y))$}
		\UnaryInfC{$\mu_{f}(p)(f(x),f(y)) \circ g(f(x)) =_{\mu(p)} g(f(y)): P(f(y))$}	
		\UnaryInfC{$ g(f(y)) =_{\sigma(\mu(p))} \mu_{f}(p)(f(x),f(y)) \circ g(f(x)) : P(f(y)) $}
		\BinaryInfC{$p(x,y) \circ (g \circ f)(x) =_{\tau(\mu(p),\sigma(\mu(p)))} \mu_{f}(p)(f(x),f(y)) \circ g(f(x)) : P(f(y))$}
		\UnaryInfC{$transport^{P \circ f}(p,u) =_{\tau(\mu(p),\sigma(\mu(p)))} transport^{P}(\mu_{f}(p),u)$}
		
\end{bprooftree}}
\end{center}

In particular, we have:

\begin{center}
		 \resizebox{1.1\hsize}{!}{	\begin{bprooftree}
		\AxiomC{$x =_{\rho_x} x :A$}
		\AxiomC{$ u \equiv (g \circ f)(x) : (P \circ f)(x)$}
  \BinaryInfC{$\rho_x(x,x) \circ (g \circ f)(x) : (P \circ f)(x)$}
		\UnaryInfC{$\rho_x(x,x) \circ (g \circ f)(x) =_{\mu_{g\circ f}(\rho_x)} (g \circ f)(x) : (P \circ f)(x)$}
		\UnaryInfC{$\rho_x(x,x) \circ (g \circ f)(x) =_{\mu_{g \circ f}(\rho_x)} g(f(x)) : P(f(x))$}
		
		\AxiomC{$x =_{\rho_x} x : A$}
		\UnaryInfC{$f(x) =_{u_{f}(\rho_x)} f(x) : B$}
		\AxiomC{$u \equiv g(f(x)) : P(f(x))$}
		\BinaryInfC{$\mu_{f}(p)(f(x),f(y)) \circ g(f(x)): P(f(x))$}
		\UnaryInfC{$\mu_{f}(\rho_x)(f(x),f(x)) \circ g(f(x)) =_{\mu_{g}(\rho_x)} g(f(x)): P(f(x))$}	
		\UnaryInfC{$ g(f(x)) =_{\sigma(\mu_{g}(\rho_x))} \mu_{f}(\rho_x)(f(x),f(x)) \circ g(f(x)) : P(f(x)) $}
		\BinaryInfC{$\rho_x(x,x) \circ (g \circ f)(x) =_{\tau(\mu_{g\circ f}(\rho_x),\sigma(\mu_{g}(\rho_x)))} \mu_{f}(\rho_x)(f(x),f(x)) \circ g(f(x)) : P(f(x))$}
		\UnaryInfC{$transport^{P \circ f}(\rho_x,u) =_{\tau(\mu_{g \circ f}(\rho_x),\sigma(\mu_{g}(\rho_x)))} transport^{P}(\mu_{f}(\rho_x),u)$}
    \end{bprooftree}}
 \end{center}
\end{proof}

\begin{lemma}
	For any $f: \Pi_{(x: A)}P(x) \rightarrow Q(x)$, $ x =_{p} y : A$ and $u(x) : P(x)$, we have:
	
	\begin{center}
		$transport^{Q}(p,f(u(x))) = f(transport^{P}(p,u(x)))$
	\end{center}
	
	\begin{proof}  The proof of this Lemma can be seen in Lemma $4.12$ of \cite{Art4}.
	\end{proof}
	
\end{lemma}
\section{Conclusion}\label{conclusao} 
Based on the idea of introducing a formal counterpart to a rewriting sequence between terms, which will then count as a proper entity and call \emph{computational path} in the syntax of type theory, we make use of such an entity to develop the central objective of our work. Using the concept of computational paths (sequences of rewrites), where these equalities reside at the level of rewrites, we show, the computational paths can be used to prove the relations obtained by the $LND_{EQ}-TRS_{2}$ writing system. Furthermore, we believe that in future work it will be possible to establish results that relate abstract algebra, algebraic topology and computational theory using the theory of computational paths as an approach.
\backmatter

\begin{appendices}

\section{LND$-{EQ}$-TRS Rules}

\ \\

Thus, we put together all those rules to compose our rewrite system:

\begin{definition}[$LND_{EQ}-TRS$ \citep{Ruy1}] 
	 \quad \\  
\begin{multicols}{2}
1. $\sigma(\rho) \triangleright_{sr} \rho$ \\ 
2. $\sigma(\sigma(r)) \triangleright_{ss} r$\\ 
3. $\tau({\mathcal C}[r] , {\mathcal C}[\sigma(r)]) \triangleright_{tr}  {\mathcal C }[\rho]$\\ 
4. $\tau({\mathcal C}[\sigma(r)], {\mathcal C}[r]) \triangleright_{tsr} {\mathcal C}[\rho]$\\ 
5. $\tau({\mathcal C}[r], {\mathcal C}[\rho]) \triangleright_{trr} {\mathcal C}[r]$\\ 
6. $\tau({\mathcal C}[\rho], {\mathcal C}[r]) \triangleright_{tlr} {\mathcal C}[r]$ \\ 
7. ${\tt sub_L}({\mathcal C}[r], {\mathcal C}[\rho]) \triangleright_{slr} {\mathcal C}[r]$\\ 
8. ${\tt sub_R}({\mathcal C}[\rho], {\mathcal C}[r]) \triangleright_{srr} {\mathcal C}[r]$ \\
9. ${\tt sub_L} ({\tt sub_L} (s, {\mathcal C}[r]), {\mathcal C}[\sigma(r)]) \triangleright_{sls} s$\\
10. ${\tt sub_L} ( {\tt sub_L} (s , {\mathcal C}[\sigma(r)]) , {\mathcal C}[r]) \triangleright_{slss} s$\\ 
11. ${\tt sub_R} ({\mathcal C}[s], {\tt sub_R} ({\mathcal C}[\sigma(s)],r)) \triangleright_{srs} r$\\ 
12. ${\tt sub_R} ({\mathcal C}[\sigma(s)], {\tt sub_R} ({\mathcal C}[s] ,  r )) \triangleright_{srrr} r$\\ 
13. 
$\mu_1 ( \xi_1 ( r))\triangleright_{mx2l1} r$\\
14. $\mu_1 ( \xi_\land ( r,s))\triangleright_{mx2l2} r$\\
15.
$\mu_2 ( \xi_\land ( r,s))\triangleright_{mx2r1} s$\\
16.
$\mu_2 ( \xi_2 ( s))\triangleright_{mx2r2} s$\\
17. 
$\mu ( \xi_1 (r) , s , u) \triangleright_{mx3l} s$\\ 
18. 
$\mu (\xi_2 (r) , s , u) \triangleright_{mx3r} u$\\ 
19.
$\nu (\xi (r)) \triangleright_{mxl} r$\\ 
20.
$\mu (\xi_2 (r) , s) \triangleright_{mxr} s$\\ 
21.
$\xi ( \mu_1 (r),\mu_2(r) ) \triangleright_{mx} r$ \\ 
22.
$\mu ( t, \xi_1 (r), \xi_2 (s)) \triangleright_{mxx} t$ \\ 
23. 
$\xi ( \nu (r) ) \triangleright_{xmr} r$ \\ 
24. 
$\mu (s,\xi_2 (r)) \triangleright_{mx1r} s$\\ 
25. $\sigma(\tau(r,s)) \triangleright_{stss} \tau(\sigma(s),  \sigma(r))$\\ 
26. $\sigma({\tt sub_L}(r,s)) \triangleright_{ssbl} {\tt sub_R}(\sigma(s), \sigma(r))$\\ 
27. $\sigma ({\tt sub_R} (r,s)) \triangleright_{ssbr} {\tt sub_L} (\sigma
(s),  \sigma (r))$\\ 
28. $\sigma(\xi (r)) \triangleright_{sx} \xi ( \sigma(r))$\\ 
29. $\sigma(\xi (s, r)) \triangleright_{sxss} \xi ( \sigma(s),  \sigma(r))$\\ 
30. $\sigma(\mu (r)) \triangleright_{sm} \mu ( \sigma(r))$ \label{rw30}\\ 
31. $\sigma(\mu (s, r)) \triangleright_{smss} \mu (\sigma(s),  \sigma(r))$\\ 
32. $\sigma(\mu (r,u,v)) \triangleright_{smsss} \mu ( \sigma(r),\sigma(u),\sigma(v))$\\
33. $\tau (r, {\tt sub_L} (\rho , s)) \triangleright_{tsbll} {\tt sub_L}  (r,s)$\\ 
34. $\tau (r, {\tt sub_R} (s, \rho)) \triangleright_{tsbrl}  {\tt 
	sub_L} (r,s)$\\ 
35. $\tau({\tt sub_L}(r,s),t) \triangleright_{tsblr} \tau (r, {\tt 
	sub_R} (s,t))$\\ 
36. $\tau ({\tt sub_R} (s,t),u) \triangleright_{tsbrr} {\tt sub_R} (s, \tau  (t,u))$\\ 
37. $\tau(\tau(t,r),s) \triangleright_{tt} \tau(t,\tau (r,s)) $\\
38. $\tau ({\mathcal C}[u], \tau ({\mathcal C}[\sigma(u)] , v)) \triangleright_{tts} v$\\
39. $\tau ({\mathcal C}[\sigma(u)] , \tau ({\mathcal C}[u] , v)) \triangleright_{tst} u$.
\end{multicols}
\label{LND-TRS}
\end{definition}




\end{appendices}


\bibliography{sn-bibliography}


\end{document}